\documentclass[12pt]{article}   
\usepackage{amsmath}
\usepackage{amssymb}   
\usepackage{amsthm}
\usepackage{setspace}
\usepackage{graphicx}	        
\usepackage{enumerate}
\usepackage[usenames,dvipsnames]{color}
\usepackage{hyperref}
\usepackage{fullpage}
\usepackage[authoryear,round]{natbib}

\newcommand{\pdfcolor}{blue}

\hypersetup{
    colorlinks=true,
    linkcolor=\pdfcolor,
    citecolor=\pdfcolor,
    urlcolor=blue}

\DeclareMathOperator*{\argmax}{argmax}
\DeclareMathOperator*{\argmin}{argmin}

\DeclareMathOperator{\marg}{marg}
\DeclareMathOperator{\supp}{supp}

\DeclareMathOperator{\Rev}{Rev}
\DeclareMathOperator{\rank}{rank}

\newcommand{\R}{\mathbb R}

\newcommand{\E}{\mathbb E}

\newcommand{\be}{\begin{equation}}
\newcommand{\ee}{\end{equation}}

\newcommand{\comment}[1]{}

\theoremstyle{plain}
\newtheorem{theorem}{Theorem}

\newtheorem{proposition}{Proposition}
\newtheorem{lemma}{Lemma}

\newtheorem*{theorem*}{Theorem}
\newtheorem*{proposition*}{Proposition}
\newtheorem*{example*}{Example}
\newtheorem*{claim*}{Claim}
\newtheorem*{lemma*}{Lemma}

\theoremstyle{definition}
\newtheorem{definition}{Definition}

\theoremstyle{remark}

\begin{document}

\title{Robust Mechanisms Under Common Valuation}

\author{Songzi Du\footnote{I thank Gabriel Carroll, Vitor Farinha Luz, Ben Golub, Wei Li, and Michael Ostrovsky for comments and discussions.  Email: \href{mailto:songzid@sfu.ca}{songzid@sfu.ca}} \\ Simon Fraser University} 

\date{November 16, 2016}

\maketitle

\begin{abstract}
We study robust mechanisms to sell a common-value good.  We assume that the mechanism designer knows the prior distribution of the buyers' common value but is unsure of the buyers' information structure about the common value.  We use linear programming duality to derive mechanisms that guarantee a good revenue among all information structures and all equilibria.  Our mechanism maximizes the revenue guarantee when there is one buyer.  As the number of buyers tends to infinity, the revenue guarantee of our mechanism converges to the full surplus.
\end{abstract}

\newpage

\section{Introduction}

In this paper we study robust mechanism design for selling a common-value good.  A robust mechanism is one that works well under a variety of circumstances, in particular under weak assumptions about participants' information structure.  The goal of robust mechanism design is to reduce the ``base of common knowledge required to conduct useful analyses of practical problems,'' as envisioned by \citet*{Wilson1987}.

The literature on robust mechanism design has so far largely focused on private value settings.\footnote{See \citet*{ChungEly}, \citet*{Brooks2013}, \citet*{FrankelAER}, \citet*{CarrollAER, Carroll_AdverseSelection}, \citet*{YamashitaRestud, YamashitaRevenueGuarantee}, \citet*{Carrasco_etal}, \citet*{ChenLi}, \citet*{HartlineRoughgarden_2016}, among others; we follow this literature by adopting a max-min approach for robust mechanisms.}  Common value is of course important in many real-life markets (particularly financial markets) and has a long tradition in auction theory.  Robustness with respect to information structure is especially relevant with common value, since there is no canonical information structure in this setting.  In practice it is hard to pinpoint exactly what is a signal (or a set of signals) for a buyer and to quantify the correlation between the signal and the common value, not to mention specifying the joint distribution of signals for all buyers that correctly captures their beliefs and higher order beliefs about the common value.

We suppose that the prior distribution of the common value is known.  We want to design a mechanism that guarantees a good revenue for every information structure consistent with the prior and every equilibrium from the information structure.  Such mechanism is clearly useful when the designer does not know the nature of information held by the strategic buyers.  Moreover, such mechanism can be adopted as the default trading protocol that works well in a variety of circumstances with minimal customization.  We are inspired by a recent paper of \citet*{BergemannBrooksMorris} which works out the revenue guarantee of the first price auction (among other results).

We come up with a mechanism that guarantees a better revenue than the first price auction.  The mechanism is simple to implement in practice and can be described as follows.  Suppose there is one common-value good to sell, and $I \geq 1$ buyers with quasi-linear utility.  Let the message space for each buyer $i$ be the interval $[0,1]$.  We think of a message $z_i \in [0,1]$ as the demand of buyer $i$.  Buyer $i$ gets the good with probability $q_i(z_i, z_{-i})$ ($q_i$ could also be buyer $i$'s quantity of allocation if the good is divisible) and pays $P_i(z_i, z_{-i})$.  If $z_1 \geq z_2 \geq \cdots \geq z_I$, then
\be
\label{eq.mechanism}
q_i(z_i, z_{-i}) = \sum_{j=i}^{I-1} \frac{z_j - z_{j+1}}{j} + \frac{z_I}{I}, \qquad P_i(z_i, z_{-i}) = X (\exp(z_i/A)-1),
\ee
and analogously for any other ordering of $(z_1, z_2, \ldots, z_I)$.  That is, the lowest buyer gets $1/I$ of his demand, the second lowest buyers gets that plus $1/(I-1)$ of the difference between his and the lowest demand, and so on.  Thus, the total probability/quantity of allocation is equal to the highest demand.  Moreover, each buyer's payment depends only on his demand and is independent of his final allocation, like an all-pay auction.  Finally, $A > 0$ and $X>0$ in Equation \eqref{eq.mechanism} are constants that are optimized for the prior distribution of value; intuitively, the constant $A$ scales the demand, while the constant $X$ scales the payment.  We call the mechanism in Equation \eqref{eq.mechanism} the \emph{exponential price mechanism}.

We prove that the exponential price mechanism gives the optimal revenue guarantee when there is one buyer ($I=1$).  In this case we know sharp upper bound on the revenue guarantee.  For example, if the prior is the uniform distribution on $[0,1]$, then the designer can guarantee (among all information structures and all equilibria) a revenue of at most $1/4$: fix any mechanism, there is the private-value information structure, and its equilibrium revenue must be less than $1/4$ which is obtained by the private-value optimal mechanism (a posted price of $1/2$).  \citet*{RoeslerSzentes} study the optimal information structure for a buyer when the seller is best responding to this information structure.  Roesler and Szentes's optimal information structure gives a subtle upper bound on the seller's revenue guarantee.   For example, when the prior is the uniform distribution on $[0,1]$, the seller can guarantee a revenue of at most $0.2036$.  We prove that the exponential price mechanism exactly guarantees the Roesler-Szentes upper bound for any prior distribution when $I=1$.\footnote{In other words, \citet*{RoeslerSzentes} characterize for one buyer:
\[
\min_{\text{info.\ structure} } \qquad \max_{\text{mechanism, equilibrium}} \quad \text{Revenue},
\]
while we characterize:
\[
\max_{\text{mechanism}} \qquad \min_{\text{info.\ structure, equilibrium}} \quad \text{Revenue},
\]
and show it is equal to their min-max value.  Equilibrium here is a mapping from signals of the information structure to messages in the mechanism, such that there is no incentive to deviate.}  In contrast, any posted price does not give the optimal revenue guarantee; for example, when the prior is the uniform $[0,1]$ distribution, the optimal posted price guarantees a revenue of only $1/8$.\footnote{\label{footnote:postedprice}When the prior is the uniform $[0,1]$ distribution, a posted price of $p \leq 1/2$ guarantees a revenue of $(1-2p)p$: suppose the buyer's information about his value is the partition $\{[0, 2p), [2p, 1] \}$, an equilibrium is to buy at price $p$ if and only if $[2p, 1]$ is realized.}


As the number $I$ of buyers increases, the exponential price mechanism guarantees a better revenue, as we numerically demonstrate in \autoref{fig.PiStar} and \autoref{tab.rev}.  We prove that as the number of buyers tends to infinity, the revenue guarantee of the exponential price mechanism (over all information structures and all equilibria) becomes arbitrarily close to the full surplus (the expectation of the common value).  Since the full surplus is an upper bound on the equilibrium revenue of every mechanism, the exponential price mechanism achieves the optimal revenue guarantee as $I \rightarrow \infty$.  This guarantee of full surplus extraction in the limit is not obtained by the first price auction (as shown by \citet*{Engelbrecht-WiggansMilgromWeber} and \citet*{BergemannBrooksMorris}), second price auction\footnote{For a second price auction with a reserve price (potentially zero), suppose there is one informed buyer who knows the common value $v$, and $I-1$ uninformed buyers who only knows the prior.  The following is an equilibrium: the informed buyer truthfully bids $v$, and all uninformed buyers bid 0.  Clearly, this equilibrium does not obtain the full surplus in revenue as $I \rightarrow \infty$.}, all-pay auction\footnote{The minimum-revenue information structure in \citet*{BergemannBrooksMorris} for the first price auction also fails to extract the full surplus in revenue for an all-pay auction as $I \rightarrow \infty$.}, or with a posted price (consider the case when all $I$ buyers have symmetric information about the common value).  And unlike the mechanisms of \citet*{CremerMcLean_1985Ecta, CremerMcLean_1988Ecta}, our mechanism is detail free and depends only on the support of the prior distribution, and extracts the full surplus for all information structures in the limit.

To study the revenue guarantee of mechanisms we introduce a duality approach which could be useful for other problems.  We want to minimize the expected revenue over the set of information structures and equilibria for a given mechanism, and then maximize the minimized revenue over the set of mechanisms.  \citet*{BergemannMorrisTE} give the powerful insight that we can combine information structure and equilibrium into a single entity called \emph{Bayes Correlated Equilibrium}, which is a joint distribution over actions and value subject to obedience and consistency constraints.  Minimizing revenue over Bayes correlated equilibria for any fixed mechanism is a linear programming problem, and we can equivalently solve the dual problem which is a maximization problem over the dual variables of constraints associated with Bayes correlated equilibrium.  These dual variables have the interpretation as transition rates for a continuous-time Markov process over the message space, similar to the transition probabilities in \citet*{Myerson1997} as dual variables for complete-information correlated equilibrium.  Moreover, we can combine the maximization over the dual variables with the maximization over the mechanism design variables, so we have a single maximization problem which is equivalent to but more tractable than the original max-min problem.



\section{Model}

\subsubsection*{Information}
The mechanism designer has a single good to sell.  Let $\mathcal{I} = \{1, 2, \ldots, I\}$ be a finite set of buyers, $I \geq 1$. The buyers have a common value $v \in V = \{0, \nu, 2\nu, \ldots, 1 \}$ for the good and have quasi-linear utility, where $\nu>0$ is a constant.  Let $p \in \Delta(V)$ be the prior distribution of common value; the prior $p$ is known by the designer as well as by the buyers.  (The designer only knows the prior $p$ about the value.)  

Each buyer $i$ \emph{may} possess some additional information $s_i \in S_i$ about the common value beyond the prior, where $S_i$ is a finite set of signals.  We have $\tilde{p} \in \Delta(V \times \prod_{i \in \mathcal{I}} S_i)$ such that $\marg_V \tilde{p}=p$,\footnote{Let $\marg_V \tilde{p}$ be the marginal distribution of $\tilde{p}$ over $V$.} so buyer $i$'s information about the common value is informed by $\tilde{p}(\, \, \cdot \mid s_i)$.  As discussed in the introduction, the information structure $(S_i, \tilde{p})_{i \in \mathcal{I}}$ is \emph{not} known by the designer.


\subsubsection*{Mechanism}

A mechanism is a set of allocation rules $q_i : M \rightarrow [0, 1]$ and payment rules $P_i : M \rightarrow \R$ satisfying $\sum_{i \in \mathcal{I}} q_i(m) \leq 1$, where $M_i$ is the message space of buyer $i$ and is a finite set, and $M = \prod_{i \in \mathcal{I}} M_i$ the space of message profiles.  A mechanism defines a game in which the buyers simultaneously submit messages and have utility
\be
\label{eq.Ui}
U_i(v, m) = v \cdot q_i(m) - P_i(m).
\ee
The allocation $q_i(m)$ can be interpreted as the probability of getting the good in the case of an indivisible good, and as the share of the good in the case of a divisible good.

We assume that a mechanism always has an opt-out option for each buyer $i$: there exists a message $m_i \equiv 0 \in M_i$ such that $q_i(0, m_{-i}) = P_i(0, m_{-i}) = 0$ for every $m_{-i} \in M_{-i}$.  

In this paper we focus on \emph{symmetric} mechanism, which satisfies
\begin{align}
\label{eq.symmetricMech}
q_i(m_i', m'_{-i}) &= q_1(m_1 = m'_i, m_{-1} = m'_{-i}) \equiv q(m_i', m'_{-i}) \\
P_i(m'_i, m'_{-i}) &= P_1(m_1 = m'_i, m_{-1} = m'_{-i}) \equiv P(m_i', m'_{-i}) \nonumber
\end{align}
for every $i \in \mathcal{I}$ and $m' \in M$.  By $m_{-1} = m'_{-i}$ we mean that $m_{-1}$ and $m'_{-i}$ have the same elements but not necessarily the same ordering of elements; for example we may have $m_{-1}=(a, b, c)$ and $m'_{-i}=(c, b, a)$.  Intuitively, in a symmetric mechanism every buyer is treated in the same way.  For a symmetric mechanism we abbreviate $q_1(m)$ to $q(m)$ and $P_1(m)$ to $P(m)$.

\subsubsection*{Equilibrium}

Given a mechanism $(q_i, P_i)_{i \in \mathcal{I}}$ and an information structure $(S_i, \tilde{p})_{i \in \mathcal{I}}$, we have a game of incomplete information.  A \emph{Bayes Nash Equilibrium} (BNE) of the game is defined by strategy $\sigma_i : S_i \rightarrow \Delta(M_i)$ for each buyer $i$ such that for every $s_i \in S_i$, the support of $\sigma_i(s_i)$ is among the best responses to others' strategies:
\be
\supp \sigma_i(s_i) \subseteq \argmax_{m_i \in M_i} \sum_{(v, s_{-i}) \in V \times S_{-i}} U_i(v, (m_i, \sigma_{-i}(s_{-i})) ) \tilde{p}(v, s_{-i} \mid s_i),
\ee
where $U_i(v, (m_i, \sigma_{-i}(s_{-i})) )$ is linearly extended from Equation \eqref{eq.Ui}.

The ex ante distribution $\mu \in \Delta(V \times M)$ generated by any BNE $(\sigma_i)_{i \in \mathcal{I}}$ of any information structure  $(S_i, \tilde{p})_{i \in \mathcal{I}}$ satisfies the following two conditions:
\begin{align}
 & \sum_{m \in M} \mu(v, m) = p(v), \quad v \in V, \label{eq.bce.consistency} \tag{Consistency} \\
 & \sum_{(v, m_{-i}) \in V \times M_{-i}} \mu(v, m) \left( U_i(v, (m_i, m_{-i})) - U_i(v, (m_i', m_{-i})) \right) \geq 0, \quad i \in \mathcal{I}, (m_i, m_i') \in M_i \times M_i. \label{eq.bce.obedience} \tag{Obedience}
\end{align}

A distribution $\mu \in \Delta(V \times M)$ that satisfies the above two conditions is called a \emph{Bayes Correlated Equilibrium} (BCE) of the mechanism $(q_i, P_i)_{i \in \mathcal{I}}$.  For any BCE $\mu$, there exists an information structure and a BNE of that information structure that generates $\mu$.  See \citet*{BergemannMorrisTE} for more details.

For notational brevity, we sometimes omit the set to which a summation variable belongs when it is obvious; for example, summing over $m$ means summing over $m \in M$.

\subsubsection*{Designer's problem}


The mechanism designer wants to solve:
\begin{align}
& \sup_{(q_i, P_i)_{i \in \mathcal{I}}} \; \min_{\mu \in \Delta(V \times M)} \; \sum_{(v, m)} \sum_{i} \mu(v, m) P_i(m) \label{eq.maxmin.problem} \\
& \text{such that $\mu$ is a BCE of $(q_i, P_i)_{i \in \mathcal{I}}$}. \nonumber
\end{align}
\begin{definition}
A mechanism \emph{guarantees} a revenue $R$ if every BCE of this mechanism has an expected revenue larger than or equal to $R$.
\end{definition}

\section{Main Results}

Our main results are a class of mechanisms that give good revenue guarantee.  Consider a symmetric mechanism with $k$ messages besides the opt-out message: $M_i=\{0, 1, \ldots, k \}$, for every buyer $i \in \mathcal{I}$.  The allocation $q(m_1, m_{-1})$ is given by:
\begin{align}
& q(0, m_{-1}) = 0,  \qquad \qquad \qquad \qquad \qquad \qquad \qquad \qquad \qquad \qquad \qquad \qquad m_{-1} \in M_{-1},
\label{eq.I.allocation} \\
& q(m_1+1, m_{-1}) - q(m_1, m_{-1}) = \left( \frac{1}{|\rank(m_1, m_{-1})|} \sum_{j \in \rank(m_1, m_{-1})}\frac{1}{j} \right) \cdot \frac{1}{k},  \quad 0 \leq m_1 \leq k-1, \nonumber
\end{align}
where $\rank(m_1, m_{-1}) \subseteq \{1, 2, \ldots, I\}$ is the set of ranks (from the top) of $m_1$ in $(m_1, m_2, \ldots, m_I)$; for example, $\rank(20, 10, 20, 40, 30) = \{3, 4\}$, because $m_1 = 20$ and $m_3 =20$ are tied for the third and the fourth place in this list; and $\rank(20, 10, 30, 40, 30) = \{4\}$ because in this list $m_1=20$ is unambiguously ranked fourth, even though there is a tie for the second and the third rank.  We think of a message $m_i$ as the demand of a fraction $m_i / k$ of the good; the allocation in Equation \eqref{eq.I.allocation} is increasing with the demand at a rate equal to the reciprocal of the demand's rank: a rate of 1 for the highest demand, of $1/2$ for the second highest demand, of $1/3$ for the third highest demand, and so on.  Moreover, we break tie in a symmetric way and randomize over all feasible ranks, in the case when $|\rank(m_1, m_{-1})| > 1$.  It is easy to check that Equation \eqref{eq.I.allocation} uniquely defines an allocation function (i.e., the feasibility condition is always satisfied); the total amount of allocation is at most $\max(m_1, m_2, \ldots, m_I)/k$.

The payment of our mechanism is:
\be
\label{eq.I.payment}
P(m_1, m_{-1}) = X \left( \left( 1+\frac{1}{a} \right)^{m_1} - 1 \right),
\ee
where $X > 0$ and $a > 0$ are constants that are optimized for a given prior distribution $p$.  That is, the payment of every buyer depends only on his message and is independent of his final allocation.

As $k \rightarrow \infty$ and $a = A \cdot k$, the mechanism from Equations \eqref{eq.I.allocation} and \eqref{eq.I.payment} converges to \eqref{eq.mechanism}, where we reparametrize $m_i \in \{0, 1, \ldots, k\}$ to $z_i \equiv m_i/k \in [0, 1]$, where $z_i$ is buyer $i$'s demand.  Thus, we abuse the terminology and refer to the mechanism from Equations \eqref{eq.I.allocation} and \eqref{eq.I.payment} as the exponential price mechanism as well.

Intuitively, the exponential price mechanism tries to be egalitarian and allocate some quantity of the good to every buyer.   Since the exponential payment is a convex function of quantity, it makes sense to split the good among all buyers.   Of course, a buyer with a higher demand gets more quantity because such buyer is paying more.  If $z_1 > z_2 > \cdots > z_I$, then buyer $i$ gets exactly $(z_i - z_{i+1})/i$ more than the allocation of buyer $i+1$; we have the factor $1/i$ because the quantity $(z_i - z_{i+1})/i$ is also acquired by all buyer $j > i+1$, and by definition there are $i$ of them.  The intuition for the exponential functional form of the payment rule is best illustrated when there is a single buyer and is presented in \autoref{sec.intuition.onebuyer}.

When there is a single buyer ($I=1$), the exponential price mechanism becomes:
\be
\label{eq.mechanism.I1}
q(m_1) = m_1/k, \qquad P(m_1)= X \left( \left( 1+\frac{1}{a} \right)^{m_1} - 1 \right), \qquad m_1 \in \{0, 1, \ldots, k\},
\ee
since $\rank(m_1) = \{1\}$ by definition.  In fact, this mechanism achieves the optimal revenue guarantee:

\begin{theorem}
\label{prop.onebuyer}
Suppose there is one buyer, and as $\nu \rightarrow 0$ the prior $p$ converges to a distribution with a positive density.  There exist constants $A > 0$ and $X > 0$ such that the exponential price mechanism with $a = A \cdot k$ and the given $X$ achieves the optimal revenue guarantee (i.e., solution to Problem \eqref{eq.maxmin.problem}) as $k \rightarrow \infty$ and $\nu \rightarrow 0$.
\end{theorem}

That is, for any $\epsilon > 0$, there exists $\bar{\nu}$ and $\bar{k}$ such that for any $\nu \leq \bar{\nu}$ and $k \geq \bar{k}$, the exponential price mechanism with $a = A \cdot k$ and the given $X$ guarantees a revenue within $\epsilon$ of the best possible from Problem \eqref{eq.maxmin.problem}.

We compute the optimal revenue guarantee of \autoref{prop.onebuyer} for various prior distributions in \autoref{tab.rev} (page \pageref{tab.rev}).

Our second result states the exponential price mechanism guarantees in expected revenue the full surplus (the expected common value) as the number of buyers tends to infinity.  In this sense the mechanism is asymptotically optimal.

\begin{theorem}
\label{prop.Iinf}
Let $a = \frac{k}{\log(I)}$ and $X = \frac{1}{2 I \log(I)}$.  The exponential price mechanism guarantees a revenue of $\sum_{v} v \cdot p(v)$ as $k \rightarrow \infty$ and $I \rightarrow \infty$.
\end{theorem}

That is, for any $\epsilon > 0$, there exists $\bar{I}$ and $\bar{k}$ such that for any $I \geq \bar{I}$ and $k \geq \bar{k}$, the exponential price mechanism with $a = \frac{k}{\log(I)}$ and $X = \frac{1}{2 I \log(I)}$ guarantees a revenue within $\epsilon$ of $\sum_{v} v \cdot p(v)$. 

The values of $a$ and $X$ in \autoref{prop.Iinf} depend only on the support of the prior (which is in $[0,1]$) and is independent of the other details of the prior.  Thus, the convergence of the mechanism's revenue guarantee to the full surplus holds for every prior supported on $[0,1]$.

We illustrate the revenue guarantee of the exponential price mechanism as a function of the number of buyers in \autoref{fig.PiStar}; we also compare with the first price auction with the reserve price chosen to maximize the revenue guarantee \citep*{BergemannBrooksMorris}.  For this figure we take the prior to be the uniform distribution on $[0,1]$ and $\nu \rightarrow 0$.  We see that the revenue guarantee of the exponential price mechanism is fairly close to the full surplus of 0.5 when there are 20 buyers.

\begin{figure}[ht!]
\centering
\includegraphics[scale=1.5]{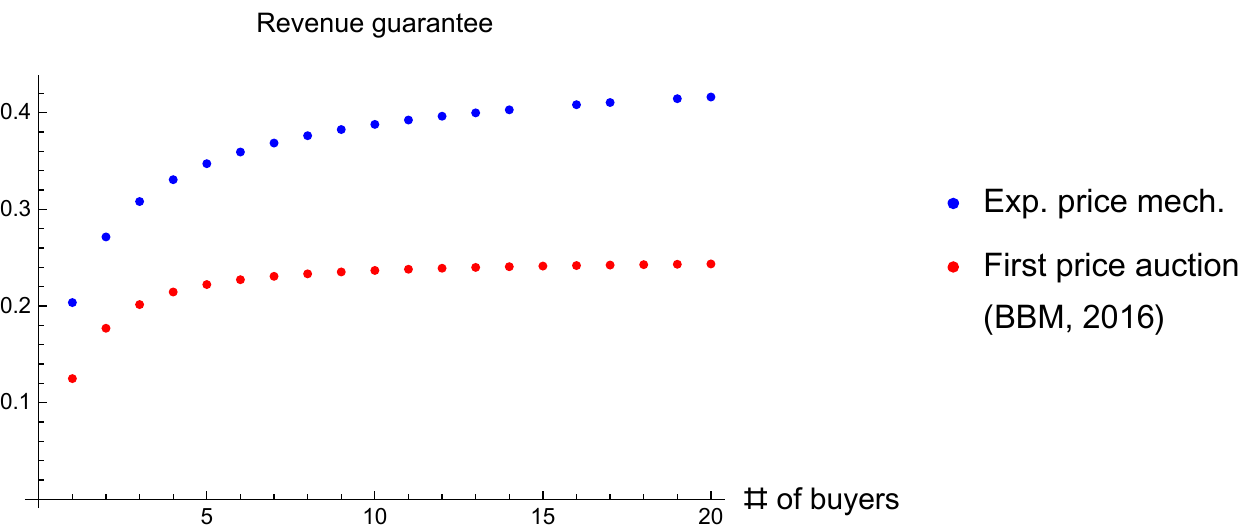}
\caption{\label{fig.PiStar}Revenue guarantees of the exponential price mechanism and of the first price auction with the optimal reserve price.}
\end{figure}


\section{Duality Approach to Robust Mechanism}
\label{sec.minRev.BCE}


To prove \autoref{prop.onebuyer} and \autoref{prop.Iinf}, we introduce a duality approach.  For a given mechanism $(q_i, P_i)_{i \in \mathcal{I}}$, the BCE that minimizes the expected revenue is found by the following problem:
\begin{align}
& \min_{\mu} \; \sum_{(v, m)} \sum_{i} P_i(m) \mu(v, m) \label{eq.minRev.primal} \\
& \text{subject to: } \nonumber \\
& \sum_{(v, m_{-i})} (U_i(v, m) - U_i(v, (m'_i, m_{-i}))) \mu(v, m) \geq 0, \quad i \in \mathcal{I}, (m_i, m_i') \in M_i \times M_i, \nonumber \\
& \sum_{m} \mu(v, m) = p(v), \quad v \in V, \nonumber \\
& \mu(v, m) \geq 0, \quad v \in V, m \in M, \nonumber
\end{align}
where $U_i$ is the utility function defined by Equation \eqref{eq.Ui}.

The dual problem to Problem \eqref{eq.minRev.primal} is:
\begin{align}
& \max_{(\alpha_i, \gamma)_{i \in \mathcal{I}}} \; \sum_{v} p(v) \gamma(v) \label{eq.minRev.dual} \\
& \text{subject to: } \nonumber \\
& \gamma(v) + \sum_i \sum_{m_i'} [U_i(v, m) - U_i(v, (m_i', m_{-i}))] \alpha_i(m_i' \mid m_i) \leq  \sum_{i} P_i(m), \quad v \in V, m \in M, \nonumber \\
& \alpha_i(m_i' \mid m_i) \geq 0,  \quad i \in \mathcal{I}, (m_i, m_i') \in M_i \times M_i, \nonumber
\end{align}
where $\alpha_i(m_i' \mid m_i)$ is the dual variable for the obedience constraint of not playing $m_i'$ when ``recommended'' to play $m_i$ in \eqref{eq.minRev.primal}, and $\gamma(v)$ is the dual variable for the consistency constraint of $\sum_{m} \mu(v, m) = p(v)$.  By the linear programming duality theorem, Problems \eqref{eq.minRev.primal} and \eqref{eq.minRev.dual} have the same optimal value; their solutions are characterized by the complementary slackness conditions.

Mechanism designer's problem in \eqref{eq.maxmin.problem} can be written as:
\begin{align}
 & \sup_{(P_i, q_i, \alpha_i, \gamma)_{i \in \mathcal{I}}} \; \sum_{v} p(v) \gamma(v) \label{eq.maxmin.problem.2} \\
& \text{subject to: } \nonumber \\
\gamma(v) \leq & \sum_{i} P_i(m) + \sum_i \sum_{m_i'} [v (q_i(m_i', m_{-i}) - q_i(m) ) - P_i(m_i', m_{-i}) + P_i(m) ] \alpha_i(m_i' \mid m_i), \quad v \in V, m \in M, \nonumber \\
& q_i(m) \geq 0, \quad \sum_{i'} q_{i'}(m) \leq 1, \quad q_i(0, m_{-i}) = P_i(0, m_{-i}) = 0, \qquad i \in \mathcal{I}, m \in M \nonumber \\ 
&  \alpha_i(m_i' \mid m_i) \geq 0,  \qquad (m_i, m_i') \in M_i \times M_i, i \in \mathcal{I}, \nonumber
\end{align}
where we label the opt-out message as $0 \in M_i$.

The advantage of problem \eqref{eq.maxmin.problem.2} over the equivalent problem \eqref{eq.maxmin.problem} is that we work with a maximization problem instead of a max-min problem.  Moreover, we work with $(\alpha_i)_{i \in \mathcal{I}}$, where each $\alpha_i$ has $|M_i \times M_i|$ dimensions, instead of $\mu$ which has $|V \times \prod_{i \in \mathcal{I}} M_i|$ dimensions; the reduction in dimensions is significant if $|V|$ is large.  Lastly, if we find a tuple $(\alpha_i, q_i, P_i)_{i \in \mathcal{I}}$ that satisfies the constraints of Problem \eqref{eq.maxmin.problem.2}, then the value of \eqref{eq.maxmin.problem.2} under such $(\alpha_i, q_i, P_i)_{i \in \mathcal{I}}$ is by definition a lower bound on the optimal revenue guarantee.  On the other hand, finding a feasible tuple $(\mu, q_i, P_i)_{i \in \mathcal{I}}$ for Problem \eqref{eq.maxmin.problem} (i.e., $\mu$ is a BCE of $(q_i, P_i)_{i \in \mathcal{I}}$) does not yield by itself any conclusion about the revenue guarantee, since there may exist another BCE $\mu'$ of $(q_i, P_i)_{i \in \mathcal{I}}$ with a lower revenue than $\mu$.

Problem \eqref{eq.maxmin.problem.2} can be summarized as:
\be
\label{eq.maxmin.problem.3}
\max_{(q_i, P_i, \alpha_i)_{i \in \mathcal{I}}} \; \sum_{v} p(v) \cdot \min_{m} \Rev(v, m),
\ee
subject to the feasibility constraints, where
\be
 \label{eq.maxmin.rev}
 \Rev(v, m) \equiv \sum_{i} \left( P_i(m) + \sum_{m_i'} ( U_i(v, m_i', m_{-i}) - U_i(v, m) ) \alpha_i(m_i' \mid m_i) \right).
 \ee
 We call $\Rev(v, m)$ the \emph{virtual revenue}.  Since $U_i(v,m)$ is a linear function of $v$, so is $\Rev(v,m)$ for any fixed $m$.  We interpret $\alpha_i(m_i' \mid m_i)$ as buyer $i$'s rate of deviation from message $m_i$ to $m_i'$, and $\Rev(v,m)$ as the revenue generated by the message profile $m$, \emph{plus} the incentive to deviate from $m$ given value $v$ and rates of deviation $(\alpha_i)_{i\in \mathcal{I}}$.  By minimizing $\Rev(v, m)$ over $m$, we are ignoring message profile $m$ that either (1) has a large revenue, or (2) there is a large incentive to deviate from $m$ by a buyer.  Intuitively, (1) and (2) combines to give equilibrium message profile with minimum revenue.


Problem \eqref{eq.maxmin.problem.2} is bounded above by $\sum_{v} v \cdot p(v)$, by the following lemma:

 \begin{lemma}
 \label{lemma.upperbound}
 For every $v \in V$, we have:
 \be
 \label{eq.FullSurplus}
\min_m \Rev(v,m) \leq v.
 \ee
 \end{lemma}

 \begin{proof}

Fix an arbitrary $v \in V$.  Consider the problem:
 \begin{align}
& \max_{\gamma, (\alpha_i)_{i \in \mathcal{I}}} \; \gamma \label{eq.CompleteInfo.problem}\\
& \text{subject to: } \nonumber \\
& \gamma + \sum_i \sum_{m_i'} (U_i(v, m) - U_i(v, (m_i', m_{-i}))) \alpha_i(m_i' \mid m_i) \leq  \sum_{i} P_i(m), \quad m \in M, \nonumber \\
& \alpha_i(m_i' \mid m_i) \geq 0,  \quad i \in \mathcal{I}, (m_i, m_i') \in M_i \times M_i. \nonumber
 \end{align}

The dual to the above problem is:
 \begin{align}
& \min_{\mu} \; \sum_{m} \mu(m) \sum_i P_i(m) \\
& \text{subject to: } \nonumber \\
& \sum_{m_{-i}} \mu(m) (U_i(v, m) - U_i(v, (m_i', m_{-i}))) \geq 0, \quad i \in \mathcal{I}, (m_i, m_i') \in M_i \times M_i, \nonumber \\
& \sum_{m} \mu(m) = 1, \nonumber \\
& \mu(m) \geq 0,  \quad m \in M, \nonumber
 \end{align}
which is minimizing the revenue over \emph{complete-information} correlated equilibria $\mu$ (for the fixed $v$).  For any $\mu$ satisfying the constraints, we have $\sum_{m_{-i}} \mu(m) U_i(v, m) = \sum_{m_{-i}} \mu(m) (v q_i(m) - P_i(m)) \geq 0$ for every $i \in \mathcal{I}$ and $m_i \in M_i$ because of the presence of the opt-out message $0 \in M_i$.  Therefore, $\sum_{m} \mu(m) \sum_i (v q_i(m) - P_i(m)) \geq 0$, and $\sum_{m} \mu(m) \sum_i P_i(m) \leq \sum_m \mu(m) \sum_i v q_i(m) \leq v$.  Thus the optimal solution of \eqref{eq.CompleteInfo.problem} is bounded above by $v$.
 \end{proof}

\subsection{A Lower Bound}
\label{sec.kmessages}
We work with symmetric mechanism $q(m) \equiv q_1(m)$ and $P(m) \equiv P_1(m)$ (cf.\ Equation \eqref{eq.symmetricMech}) and symmetric $\alpha(m_i' \mid m_i) \equiv \alpha_i(m_i' \mid m_i)$.

Instead of directly solving Problem \eqref{eq.maxmin.problem.3}, we make some educated guess on $(q_i, P_i, \alpha_i)$ to get a lower bound on the maximum value of Problem \eqref{eq.maxmin.problem.3}.  Suppose $M_i = \{0, 1, \ldots, k \}$ for every buyer $i \in \mathcal{I}$, where $q(0, m_{-1}) = 0 = P(0, m_{-1})$ for every $m_{-1} \in M_{-1}$.

We focus on
\be
\alpha(j' \mid j) =
\begin{cases}
a & j' = j+1 \\
0 & j' \neq j+1
\end{cases}, \quad (j, j') \in \{0, 1, \ldots, k \}^2.
\label{eq.kmessages.condAlpha}
\ee
Condition \eqref{eq.kmessages.condAlpha} says that the local obedience constraint in BCE is binding: if the above $\alpha$ satisfies the complementarity slackness condition with a BCE $\mu$, then a buyer is indifferent between messages $j$ and $j+1$ if he is ``recommended'' to submit $j$ in the BCE $\mu$.  This is a discrete analogue of the first order condition at $j$.

Condition \eqref{eq.kmessages.condAlpha} implies that there are two kinds of messages: the ``interior'' message $j \in \{0, 1, \ldots, k-1\}$, and the ``boundary'' message $j = k$.  Thus there are $I+1$ kinds of message profiles $m \in M = \{0, 1, \ldots, k\}^I$, depending on the number of boundary messages in $m$.  For $0 \leq n \leq I$, define the class of message profiles:
\be
M(n) = \{ m \in M : |\{i \in \mathcal{I} : m_i = k\}| = n \}.
\ee
The sets $M(n)$, $0 \leq n \leq I$, form a partition of $M$.  Our second assumption is that
\be
\label{eq.kmessages.condRev}
\Rev(v, m) = \Rev(v, m') \; \forall v \in V, \qquad \text{if $m$ and $m'$ belong to the same $M(n)$}.
\ee

Condition \eqref{eq.kmessages.condRev} attemps to make $\Rev(v, m)$ over $m$ as redundant as possible, to minimize the number of items inside the $\min$ operator in Equation \eqref{eq.maxmin.problem.3}.

We now go to the exponential price mechanism defined by Equations \eqref{eq.I.allocation} and \eqref{eq.I.payment}.  Clearly, if there are $n$ boundary messages in a message profile $m$, then the rest (the interior messages) have ranks among $\{n+1, n+2, \ldots, I\}$, and by Equation \eqref{eq.I.allocation} we have:
\be
\sum_{i : m_i < k} q(m_i+1, m_{-i}) - q(m_i, m_{-i}) = \frac{1}{k} \sum_{j=n+1}^I \frac{1}{j}.
\ee
Moreover, for an interior $m_i$ we have:
\be
P(m_i, m_{-i}) - a (P(m_i+1, m_{-i})-P(m_i, m_{-i})) = -X
\ee
by Equation \eqref{eq.I.payment}.\footnote{\label{footnote.P.diffEq} In fact, Equation \eqref{eq.I.payment} is the solution to the difference equation
\[
P(m_i, m_{-i}) - a (P(m_i+1, m_{-i})-P(m_i, m_{-i})) = -X
\]
 for $m_i \in \{0, 1, \ldots, k-1\}$, with the initial condition of $P(0, m_{-i})=0$.}  Therefore, under Condition \eqref{eq.kmessages.condAlpha} we have:
\be
\Rev(v, m) = \frac{a \, v}{k} \sum_{j=n+1}^I \frac{1}{j} + n X ((1+1/a)^k - 1) - (I-n) X, \quad \text{if $m \in M(n)$.}
\ee
Thus, Condition \eqref{eq.kmessages.condRev} is satisfied, and we have the following lower bound on the maximum value of Problem \eqref{eq.maxmin.problem.3}:
\be
\label{eq.kmessages.problem}
\Pi^*_I \equiv \sup_{k \geq 1, \, a \geq 0, \, X} \left( \sum_v p(v) \cdot \min_{0 \leq n \leq I} \left( \frac{a \, v}{k} \sum_{j=n+1}^I \frac{1}{j} + n X ((1+1/a)^k - 1) - (I-n) X \right)  \right).
\ee


\begin{proposition}
\label{prop.kmessages}
The exponential price mechanism guarantees a revenue of $\Pi^*_I$ defined in \eqref{eq.kmessages.problem}.  
\end{proposition}
\begin{proof}
The proof is given by the construction above.
\end{proof}

We prove \autoref{prop.onebuyer} and \autoref{prop.Iinf} by studying $\Pi^*_1$ and $\lim_{I \rightarrow \infty} \Pi^*_I$.  In \autoref{fig.PiStar} (page \pageref{fig.PiStar}) we plot $\Pi_I^*$ for the uniform $[0,1]$ distribution, as $\nu \rightarrow 0$.


\subsection{Proof of \autoref{prop.onebuyer}}

Let $I=1$.  Suppose as $\nu \rightarrow 0$, the prior $p$ converges to a distribution with density $\rho$, where $\rho: [0,1] \rightarrow [0, \infty)$ is positive almost everywhere.

 To prove \autoref{prop.onebuyer}, we need to discuss a revelant result in \citet*{RoeslerSzentes}.  \citet*{RoeslerSzentes} study the optimal information structure for the buyer (and the worst for the seller) when the seller best responds to the information structure.  Such information structure has the following cumulative distribution function for the signals:
\be
\label{eq.RoeslerSzentes.signals}
G_{\pi}^B(s) =
\begin{cases}
1 & s \geq B \\
1 - \pi/s & s \in [\pi, B) \\
0 & s < \pi
\end{cases},
\ee
where $s \in [0,1]$ is an unbiased signal of the buyer for his value ($\E[v \mid s] = s$), $0 < \pi \leq B$ are two free parameters, and there is an atom of size $\pi/B$ at $s=B$.  If the buyer has this distribution of unbiased signals and observes the realization of the signal, then the seller is clearly indifferent between every posted price in $[\pi, B]$ and has a revenue of $\pi$ from the optimal mechanism (which is a posted price).\footnote{Intuitively, if the seller has a strict incentive over the posted price, then we can slightly change the buyer's information structure to lower the seller's optimal revenue and to increase the buyer's surplus, while preserving the seller's best response in posted price.}  Thus, $\pi$ is an upper bound on the seller's revenue guarantee.

Given the density $\rho(v)$, $G_{\pi}^B(s)$ is a distribution of an unbiased signal on $v$ if and only if $\rho$ is a mean-preserving spread of $G_{\pi}^B(s)$, which holds if and only if:
\begin{align}
\label{eq.logB}
& \int_{0}^{1} v \, \rho(v) \, dv = \int_{0}^{1} s \, dG_\pi^B(s) = \pi + \pi \log B - \pi \log \pi \\
\label{eq.SOSD}
& \min_{s \in [\pi, B]} F(s, \pi) \geq 0, \text{ where } \\
& F(s, \pi) \equiv  \int_{s'=0}^{s} \int_{v=0}^{s'} \rho(v) \,dv \, ds' - \int_{0}^{s} G_\pi^B(s') \, ds' = \int_{s'=0}^{s} \int_{v=0}^{s'} \rho(v) \, dv \, ds' - (s - \pi - \pi \log s  + \pi \log \pi), \nonumber
\end{align}
i.e., $G_\pi^B$ has the same mean as $\rho$ and second-order stochastically dominates $\rho$.  Let $B=B(\pi)$ be defined from $\pi$ by Equation \eqref{eq.logB}.

\citet*{RoeslerSzentes} prove that the best information structure for the buyer (and the worst for the seller) when the seller best responds to the information structure is $G_{\pi^*}^{B^*}$, where $\pi^*$ is the smallest $\pi$ such that $\min_{s \in [\pi, B(\pi)]} F(s, \pi) \geq 0$, and $B^* \equiv B(\pi^*)$; that is, $\pi^*$ is the smallest $\pi$ such that $\rho$ is a mean-preserving spread of $G_{\pi}^{B(\pi)}(s)$.  For our purpose, by making $\pi$ small we tighten the upper bound on the seller's revenue guarantee.

We now show that the exponential price mechanism can obtain the upper bound $\pi^*$.  In the case of one buyer, Problem \eqref{eq.kmessages.problem} simplifies to:
\be
\label{eq.onebuyer.problem}
\Pi^*_1 \equiv \max_{k \geq 1, \, a \geq 0, \, X} \; \sum_{v} \min \left( \frac{a \, v}{k} - X, \, X \left( \left( 1+\frac{1}{a} \right)^k - 1 \right) \right) \, p(v),
\ee


Set $a = A \cdot k$.  As $\nu \rightarrow 0$ and $k \rightarrow \infty$, we have
\be
\label{eq.onebuyer.uniform.revenue}
\sum_{v} \min( a v / k , X ( 1 + 1/a )^k ) \, p(v) - X \longrightarrow \Pi_1 \equiv \int_{0}^{1} \min( A v, X \exp(1/A) ) \, \rho(v) dv -X.
\ee

We maximize $\Pi$ over $A$ and $X$.  Suppose $\frac{X \exp(1/A)}{A} \in [0,1]$, the first order condition is:
\begin{align}
\label{eq.onebuyer.solution.2}
& \frac{\partial \Pi_1}{\partial X} = \int_{\frac{X \exp(1/A)}{A}}^{1} \exp(1/A) \, \rho(v) \, dv - 1 = 0, \\
& \frac{\partial \Pi_1}{\partial A} = \int_{0}^{\frac{X \exp(1/A)}{A}} v \, \rho(v) \, dv - \int_{\frac{X \exp(1/A)}{A}}^{1} \frac{X \exp(1/A)}{A^2} \, \rho(v) \, dv = 0. \nonumber
\end{align}
If the above first order condition holds and $\frac{X \exp(1/A)}{A} \in [0,1]$, then we have $\Pi_1 = X/A$.

Going back to the construction of Roesler-Szentes, let $s^*$ be an arbitrary selection from $\argmin_{s \in [\pi^*, B^*]} F(s, \pi^*)$.  Since $\min_{s \in [\pi, B(\pi)]} F(s, \pi)$ is a continuous function of $\pi$, we have $F(s^*, \pi^*) = 0$.  Moreover, $s^*$ must be interior\footnote{If $B^* < 1$, we must have $F(B^*, \pi^*)> 0$, for otherwise we would have $\int_{0}^1 G_{\pi^*}^{B^*}(s) \, ds > \int_{s=0}^1 \int_{v=0}^s \rho(v) \, dv \, ds$, which would contradict the fact that $G_{\pi^*}^{B^*}$ has the same mean as $p$.  If $B^* = 1$, then we have $\frac{\partial F}{\partial s}(B^*, \pi^*) = \frac{\pi^*}{B^*} > 0$.  In any case $s^* \neq B^*$.  Since $F(\pi^*, \pi^*) > 0$, we also have $s^* \neq \pi^*$.}, so we have $\frac{\partial F}{\partial s}(s^*, \pi^*) = 0$.


Therefore, we have (the first line is $\frac{\partial F}{\partial s}(s^*, \pi^*) = 0$, and the second line is $F(s^*, \pi^*) = 0$):
\begin{align}
\label{eq.RS.conditions}
& \int_{v=0}^{s^*} \rho(v) \, dv - 1 + \pi^*/s^* = 0, \\
& \int_{s=0}^{s^*} \int_{v=0}^s \rho(v) \, dv \, ds - (s^* - \pi^* - \pi^* \log s^*  + \pi^* \log \pi^*) = - \int_{0}^{s^*} v \, \rho(v) \, dv + \pi^* \log s^* - \pi^* \log \pi^* = 0, \nonumber
\end{align}
where in the second equality of the second line we use integration by parts and substitute in the first line.  Clearly, there exist unique $A>0$ and $X>0$ such that $s^* = X \exp(1/A)/A$ and $\pi^* = X/A$.  (We have $s^* < B^* < 1$.) Then the above equations become:
\be
\int_{\frac{X \exp(1/A)}{A}}^{1} \rho(v) \, dv = \exp(-1/A), \qquad \int_{0}^{\frac{X \exp(1/A)}{A}} v \, \rho(v) \, dv -\frac{X}{A^2} = 0, \nonumber
\ee
which is clearly equivalent to Equation \eqref{eq.onebuyer.solution.2}.

Therefore, for any $\epsilon > 0$, when $k$ is sufficiently large and $\nu$ sufficiently small, we have $\Pi_1^* \geq \pi^* - \epsilon$.  When $\nu$ is sufficiently small, the revenue guarantee must be smaller than $\pi^* + \epsilon$ by the Roesler-Szentes construction.  This concludes the proof.

\subsubsection{Intuition on \autoref{prop.onebuyer}.}
\label{sec.intuition.onebuyer}
We note that as $k \rightarrow \infty$ and $a = A \cdot k$, the exponential price mechanism in Equation \eqref{eq.mechanism.I1} becomes:
\be
q(z) = z, \qquad P(z) = X (\exp(z/A) - 1),
\ee
where $z \equiv m_1/k \in [0, 1]$ is the demand of the buyer.

Fix an unbiased information structure $(S, G)$ for the buyer: $\E[v \mid s] =s$ for every $s \in S \subseteq [0,1]$, and $s \in S$ has the cumulative distribution function $G(s)$.

Given a realization of signal $s$, the buyer solves:
\[
\max_{z} s \cdot z - X (\exp(z/A) - 1).
\]
If $s \leq \pi^* \equiv X/A$, then the buyer's optimal demand is $z=0$, and he pays 0; if $s \geq s^* \equiv X \exp(1/A)/A$, then the buyer's optimal demand is $z=1$, and he pays $X (\exp(1/A) - 1))$.  If $s \in [\pi^*, s^*]$, then the optimal demand is given by the first order condition $s = X \exp(z/A)/A$, and the buyer pays $X (\exp(z/A) - 1) \rvert_{s = X \exp(z/A)/A} = A s - X$.  Thus, the equilibrium revenue under the unbiased information structure $(S, G)$ is:
\be
\Pi_1(G) \equiv \int_{s = \pi^*}^{1} \min(A s - X, X (\exp(1/A) - 1)) \, dG(s),
\ee
which is similar to $\Pi_1$ in Equation \eqref{eq.onebuyer.uniform.revenue}, but with a different lower limit in the integral.

In general, we have:
\be
\label{eq.Pi.G}
\Pi_1(G) \geq \int_{s = 0}^{1} \min(A s - X, X (\exp(1/A) - 1)) \, dG(s) \geq \int_{s = 0}^{1} \min(A s - X, X (\exp(1/A) - 1)) \, \rho(s) \, ds = \Pi_1,
\ee
since $\rho$ is a mean-preserving spread of $G$, and $\min(A s - X, X (\exp(z/A) - 1))$ is a concave function of $s$; thus, $\Pi_1$ is a lower bound on the equilibrium revenue over all information structures, confirming \autoref{prop.kmessages} when $I=1$.

The proof of \autoref{prop.onebuyer} shows that $\Pi_1=\Pi_1(G)$ when $G = G_{\pi^*}^{B^*}$ as constructed by Roesler and Szentes.  This can be seen in Equation \eqref{eq.Pi.G} as follows: the first inequality in \eqref{eq.Pi.G} is an equality when $G = G_{\pi^*}^{B^*}$ because $G_{\pi^*}^{B^*}$ is supported on the interval $[\pi^*, B^*]$; the second inequality in \eqref{eq.Pi.G} is an equality when $G = G_{\pi^*}^{B^*}$ because $\int_{s=\pi^*}^{s^*} s \, dG_{\pi^*}^{B^*}(s) = \int_{v=0}^{s^*} v \, \rho(v) \, dv$ and $G_{\pi^*}^{B^*}(s^*) = \int_{v=0}^{s^*} \rho(v) \, dv $ by Equation \eqref{eq.RS.conditions}.


\subsection{Proof of \autoref{prop.Iinf}}

For each $0 \leq n \leq I$, define
\be
\Rev_n(v) = \frac{a \, v}{k} \sum_{j=n+1}^I \frac{1}{j} + n X ((1+1/a)^k - 1) - (I-n) X,
\ee
which is $\Rev(v, m)$ for any $m \in M(n)$.  Let
\be
v(n) = \frac{(n+1) k X}{a} (1+1/a)^k,
\ee
 By construction, we have $\Rev_n(v(n)) = \Rev_{n+1}(v(n))$ for each $0 \leq n \leq I-1$.  Set $v(-1) = 0$ and $v(I)=\infty$.  Clearly, if $X>0$, then $\Rev_n(v) = \min_{0 \leq n' \leq I} \Rev_{n'}(v)$ if and only if $v \in [v(n-1), v(n)]$.

We want to approximate the identity function $v$ by $\min_{0 \leq n \leq I} \Rev_n(v)$.  To do so, we set $A = 1/\log(I)$, $X = 1/(2 I \log(I))$ and $a = A k$.  We have $\lim_{k \rightarrow \infty} (1+1/a)^k = I$, $\lim_{k \rightarrow \infty} v(1) = 1$.  Thus,
\begin{align}
\lim_{k \rightarrow \infty} \, \sum_{v \in V} p(v) \min_{0 \leq n \leq I} \Rev_{n}(v) = & \sum_{v \leq v(0), v \in V} p(v) \left( \frac{ v}{\log(I)} \sum_{j=1}^I \frac{1}{j} - \frac{1}{2 \log(I)} \right) \\
 & + \sum_{v > v(0), v \in V} p(v) \left( \frac{ v}{\log(I)} \sum_{j=2}^I \frac{1}{j} - \frac{1}{\log(I)} \right) \nonumber
\end{align}
Clearly, the above equation converges to $\sum_{v} v \cdot p(v)$ as $I \rightarrow \infty$.  This completes the proof.

\section{Generalization}
\label{sec.twobuyers}

Since they play a central role in the derivation of exponential price mechanism, we devote this section to better understand Conditions \eqref{eq.kmessages.condAlpha} and \eqref{eq.kmessages.condRev}.  We first show that to maximize the revenue guarantee it is without loss of generality to assume Condition \eqref{eq.kmessages.condAlpha}.  Given Condition \eqref{eq.kmessages.condAlpha}, we then fully characterize the implications of Condition \eqref{eq.kmessages.condRev}, which give a generalization of exponential price mechanism.


\subsection{Simplifying the Dual Variables}

In this subsection we show that, as long as $k$ is large, for the optimal revenue guarantee it is without loss of generality to focus on symmetric mechanism with $M_i = \{0, 1, \ldots, k \}$ and dual variables in Equation \eqref{eq.kmessages.condAlpha}, for every $i \in \mathcal{I}$.


\subsubsection*{Transition probability matrix}

Fix an arbitrary tuple $(q_i, P_i, \alpha_i)_{i \in \mathcal{I}}$.  Recall that there exists an opt-out message $0 \in M_i$ for every buyer $i$.  We define a new $(\tilde{q}_i, \tilde{P}_i, \tilde{\alpha}_i)_{i \in \mathcal{I}}$ as follows:

\begin{enumerate}[(i)]
\item Define
\be
a \equiv \max_{i \in \mathcal{I}} \max_{m_i \in M_i} \; \sum_{m_i' \neq m_i} \alpha_i(m_i' \mid m_i) + c,
\ee
where $c > 0$ is an arbitrary constant, and
\be
\mathcal{A}_i(m_i' \mid m_i) \equiv
\begin{cases}
\alpha_i(m_i' \mid m_i) / a & m_i' \neq m_i \\
1 - \sum_{m_i'' \neq m_i} \mathcal{A}_i(m_i'' \mid m_i) & m_i' = m_i
\end{cases}, \qquad (m_i, m_i') \in M_i^2
\ee
If we interpret $\alpha_i$ as the transition rates of a continuous-time Markov process over $M_i$, then $\mathcal{A}_i$ is a transition probability matrix embedded in the process.  Because of the positive constant $c$, $\mathcal{A}_i(m_i \mid m_i) > 0$ for every $m_i \in M_i$, so every $m_i$ is aperiodic in $\mathcal{A}_i$ (\citet*{StroockMarkov}, Section 3.1.3)

\item Define a new message space $\tilde{M}_i \equiv \{0, 1, \ldots, k\}$ for every buyer $i$.  Message 0 is still the opt-out message. Message $j \in \tilde{M}_i$ is the ``mixed-strategy'' message given by $(\mathcal{A}_i)^j(\, \cdot \mid 0) \in \Delta(M_i)$:
\begin{align}
\tilde{q}_i(j, m_{-i}) &\equiv \sum_{m_i \in M_i} (\mathcal{A}_i)^j(m_i \mid 0) \, q_i(m_i, m_{-i}), \qquad m_{-i} \in M_{-i}, \\
\tilde{P}_i(j, m_{-i}) &\equiv \sum_{m_i \in M_i} (\mathcal{A}_i)^j(m_i \mid 0) P_i(m_i, m_{-i}), \notag
\end{align}
where $(\mathcal{A}_i)^j$ is $\mathcal{A}_i$ raised to the $j$-th power (the $j$-step transition probability matrix).  Moreover, we extend $\tilde{q}_i(\tilde{m}_i, m_{-i})$ and $\tilde{P}_i(\tilde{m}_i, m_{-i})$ linearly to $\tilde{q}_i(\tilde{m}_i, \tilde{m}_{-i})$ and $\tilde{P}_i(\tilde{m}_i, \tilde{m}_{-i})$ for each $\tilde{m}_{-i} \in \tilde{M}_{-i}$.

\item Define
\be
\tilde{\alpha}_i(\tilde{m}_i' \mid \tilde{m}_i) \equiv
\begin{cases}
a & \tilde{m}_i' = \tilde{m}_i+1, \\
0 & \tilde{m}_i' \neq \tilde{m}_i+1.
\end{cases}
\ee

\end{enumerate}

Let $\tilde{U}_i(v, \tilde{m}) \equiv v \cdot \tilde{q}_i(\tilde{m}) - \tilde{P}_i(\tilde{m})$.  We have:
\begin{align}
& \sum_{m_i \in M_i} (\mathcal{A}_i)^j(m_i \mid 0) \sum_{m_i' \in M_i} [ U_i(v, (m_i', m_{-i}))  - U_i(v, (m_i, m_{-i})) ] \alpha_i(m_i' \mid m_i) \\
& = \sum_{m_i \in M_i} (\mathcal{A}_i)^j(m_i \mid 0) \sum_{m_i' \in M_i} a \cdot [ U_i(v, (m_i', m_{-i}))  - U_i(v, (m_i, m_{-i})) ] \mathcal{A}_i(m_i' \mid m_i) \notag \\
& = a \sum_{m_i' \in M_i} \left[ U_i(v, (m_i', m_{-i})) \, (\mathcal{A}_i)^{j+1}(m_i' \mid 0) - U_i(v, (m_i', m_{-i})) \, (\mathcal{A}_i)^j(m_i' \mid 0) \right] \notag \\
& = a \, [\tilde{U}_i(v, (j+1, m_{-i})) - \tilde{U}_i(v, (j, m_{-i}))] \notag \\
& \begin{cases}
= \tilde{\alpha}_i ( j+1 \mid j) \, [\tilde{U}_i(v, (j+1, m_{-i})) - \tilde{U}_i(v, (j, m_{-i}))] & j < k, \\
\leq \epsilon & j = k,
\end{cases} \notag
\end{align}
where in the last line, $k$ is chosen to be sufficiently large so that $a \, || (\mathcal{A}_i)^{k+1}(\, \cdot \mid 0) - (\mathcal{A}_i)^{k}(\, \cdot \mid 0) ||_{\text{v}} \cdot \max_{v, m} |U_i(v, m)| \leq \epsilon$, where $|| \cdot ||_{\text{v}}$ is the total variation norm, and $\lim_{k \rightarrow \infty} (\mathcal{A}_i)^{k}$ is well defined because every $m_i$ is aperiodic (\citet*{StroockMarkov}, Section 4.1.7).

Consequently, the virtual revenue $\tilde{\Rev}(v, \tilde{m})$ given by $(\tilde{P}_i, \tilde{q}_i, \tilde{\alpha}_i)_{i \in \mathcal{I}}$ satisfies
\be
\tilde{\Rev}(v, \tilde{m})
\begin{cases}
= \sum_{m \in M} \prod_i (\mathcal{A}_i)^{\tilde{m}_i}(m_i \mid 0) \cdot \Rev(v, m) \geq \min_{m \in M} \Rev(v, m) & \tilde{m}_i < k, \forall i \in \mathcal{I}, \\
\geq \sum_{m \in M} \prod_i (\mathcal{A}_i)^{\tilde{m}_i}(m_i \mid 0) \cdot \Rev(v, m) - I \cdot \epsilon \geq \min_{m \in M} \Rev(v, m) - I \cdot \epsilon & \text{otherwise}. \\
\end{cases}
\ee

\subsubsection*{Symmetric mechanism}

Now fix a mechanism $(q_i, P_i)_{i \in \mathcal{I}}$ such that $M_i = \{0, 1, \ldots, k \}$, for every $i \in \mathcal{I}$, and suppose Equation \eqref{eq.kmessages.condAlpha} holds for $\alpha_i$.  We symmetrize the mechanism:
\begin{align}
\tilde{q}_1(m_1', m_{-1}') &\equiv \frac{1}{I!} \sum_{\sigma} q_{\sigma(1)}(m_{\sigma(1)} = m_1', m_{\sigma(2)} = m_2', \ldots, m_{\sigma(I)} = m_I'), \qquad m' \in M \\
\tilde{P}_1(m_1', m_{-1}') &\equiv \frac{1}{I!} \sum_{\sigma} P_{\sigma(1)}(m_{\sigma(1)} = m_1', m_{\sigma(2)} = m_2', \ldots, m_{\sigma(I)} = m_I'), \notag
\end{align}
where we sum over all permutations $\sigma$ over $\{1, 2, \ldots, I\}$. Likewise for the allocations and payments of buyers $2, 3, \ldots, I$.  Since $\alpha_i$ is symmetric, the virtual revenue $\tilde{\Rev}(v, m)$ given by $(\tilde{q}_i, \tilde{P}_i, \alpha_i)_{i \in \mathcal{I}}$ is the average of $\Rev(v, (m_{\sigma(1)}, m_{\sigma(2)}, \ldots, m_{\sigma(I)}))$ over all permutations $\sigma$ and given by $(q_i, P_i, \alpha_i)_{i \in \mathcal{I}}$; thus, we have $\tilde{\Rev}(v, m) \geq \min_{m} \Rev(v, m)$.

\subsubsection*{Implications for Numerical Solution}
For a fixed value of $a$ in Equation \eqref{eq.kmessages.condAlpha}, maximizing the revenue guarantee over symmetric mechanism $(q, P)$ is a linear programming problem (cf.\ Problem \eqref{eq.maxmin.problem.2}).  By varying the one-dimensional variable $a$ and solving the corresponding linear programming problems, we have a tractable method to numerically solve for the optimal revenue guarantee.


\subsection{Generalized Exponential Price Mechanism}

Given Condition \eqref{eq.kmessages.condAlpha} on dual variables (which is without loss of generality by the previous subsection), Condition \eqref{eq.kmessages.condRev} is a natural assumption on the virtual revenue.  In this subsection we fully characterize the implications of Condition \eqref{eq.kmessages.condRev} for the revenue guarantee; from the characterization we arrive at a generalization of the exponential price mechanism.  We numerically demonstrate that the revenue guarantee of the exponential price mechanism is quite good within this class of generalization.  We also compare with the revenue guarantee of the first price auction and with some upper bound on the revenue guarantee.

For simplicity, suppose $I=2$.  We consider a symmetric mechanism with $k+1$ messages: $M_1=M_2=\{0, 1, \ldots, k \}$.  Our usual assumption on the mechanism is:
\be
\label{eq.kmessages.UsualConditions}
q(0, j) = 0 = P(0, j), \quad q(j,l) \geq 0, \quad q(j, l) + q(l, j) \leq 1,  \quad (j, l) \in \{0, 1, \ldots, k\}^2.
\ee


Assume Condition \eqref{eq.kmessages.condAlpha}. Condition \eqref{eq.kmessages.condRev} holds under $I=2$ if and only if
\begin{align}
\label{eq.kmessages.qsystem}
& 2 q(1,0) = q(j+1,l)-q(j,l) + q(l+1, j) - q(l, j), \quad (j, l) \in \{0, 1, \ldots, k-1\}^2, \nonumber \\
& q(1,k) = q(j+1,k)-q(j,k), \quad j \in  \{0, 1, \ldots, k-1\},
\end{align}
and
\begin{align}
\label{eq.kmessages.Psystem}
 - 2 a P(1,0) = P(j,l)+P(l,j) - a ( P(j+1,l) - P(j,l)) - a ( P(l+1, j) - P(l, j)), \qquad & \nonumber \\
 (j, l) \in \{0, 1, \ldots, k-1\}^2, & \nonumber \\
 P(k,0) - a P(1,k) = P(j,k) + P(k,j) - a (P(j+1,k)-P(j,k)), \quad  j \in  \{0, 1, \ldots, k-1\}. &
\end{align}

It is without loss to assume that the feasibility constraint $q(j,k) + q(k, j) \leq 1$ binds for every $j$ (if not, we can increase $q(k, j)$, which strictly increases $\Rev(v, (k-1, j))$, without decreasing any other $\Rev(v, m)$ or violating any feasibility constraint):
\be
\label{eq.kmessages.q.BindingConstr}
q(j,k) + q(k, j) = 1, \quad j \in \{0, 1, \ldots, k\}.
\ee

\begin{lemma}
\label{lemma.kmessages.q}
For any allocation $q$ that satisfies Conditions \eqref{eq.kmessages.qsystem} and \eqref{eq.kmessages.q.BindingConstr}, we have $q(1,0) = (3k+1)/(4 k^2)$ and $q(1,k) = 1 / (2 k)$.
\end{lemma}

\begin{lemma}
\label{lemma.kmessages.P}
For any payment $P$ that satisfies Condition \eqref{eq.kmessages.Psystem}, we have:
\be
\label{eq.kmessages.Pkk}
P(k,k) = \left( (1 + 1/a)^k - 1 \right)^2 a P(1,0) + \left((1 +1/a)^k - 1 \right) (a P(1,k) - P(k,0) ).
\ee
\end{lemma}


Define,
\be
\label{eq.kmessages.XY}
Y_0 \equiv 2 a P(1,0), \qquad Y_1 \equiv -P(k, 0) + a P(1, k),
\ee
i.e., $-Y_0$ is equal to the first line of \eqref{eq.kmessages.Psystem}, and $-Y_1$ is equal to the second line of \eqref{eq.kmessages.Psystem}.

\autoref{lemma.kmessages.q} and \autoref{lemma.kmessages.P} lead us unambiguously to the following problem:
\be
\label{eq.twobuyers.problem}
\Pi^\#_2 \equiv \sup_{k \geq 1, \, a \geq 0, \, Y_0, \, Y_1} \; \sum_v \min \left( \frac{3k+1}{2 k^2} a v - Y_0, \, \frac{a v}{2 k} - Y_1, \, Y_0 \left( (1 + 1/a)^k - 1 \right)^2 + 2 Y_1 \left( (1 +1/a)^k-1 \right)  \right)  p(v).
\ee


Comparing $\Pi_2^{\#}$ above with $\Pi_2^*$ in Equation \eqref{eq.kmessages.problem}, the main difference is that in $\Pi^\#_2$ there are two variables $Y_0$ and $Y_1$, instead of a single variable $X$ in $\Pi_2^*$; the difference between the coefficient of $\frac{3k+1}{2 k^2}$ in $\Pi_2^{\#}$ and of $\frac{3}{2 k}$ in  $\Pi_2^*$ is unimportant, since $k \rightarrow \infty$ in both maximization problems.  In fact, if $k \rightarrow \infty$, $\Pi_2^*$ is a special case of $\Pi_2^{\#}$ with $Y_0 = 2 X$ and $Y_1 = X - X ((1+1/a)^k-1)$.  Thus, we have $\Pi_2^{\#} \geq \Pi_2^*$ as $k \rightarrow \infty$.

\begin{proposition}
\label{prop.twobuyers}
Suppose there are two buyers.  There exists a symmetric mechanism that guarantees a revenue of $\Pi^\#_2$ defined in \eqref{eq.twobuyers.problem}.  
\end{proposition}

We now specify the mechanism for \autoref{prop.twobuyers}.  Consider the following allocation rule:
\begin{align}
\label{eq.kmessages.q}
& q(0, l) = 0, \qquad l \in \{ 0, 1, \ldots, k \}, \\
& q(j+1, l)-q(j,l) =
\begin{cases}
(2k+1) / (4 k^2) & j < l \\
(3k + 1)/(4 k^2) & j = l \\
(4k+1) / (4 k^2) & j > l
\end{cases}, \qquad (j, l) \in \{ 0, 1, \ldots, k-1\}^2,  \\
& q(j+1, k)-q(j,k) = 1 / (2k), \qquad j \in \{ 0, 1, \ldots, k-1 \}. \nonumber
\end{align}
It is easy to check that the above allocation rule satisfies the feasibility constraint, and Conditions \eqref{eq.kmessages.qsystem} and \eqref{eq.kmessages.q.BindingConstr}.  As $k \rightarrow \infty$, the above allocation rule becomes identical to the allocation rule of exponential price mechanism in Equation \eqref{eq.I.allocation}.


Given any values of $Y_0$ and $Y_1$, we can choose the following solution to Equation \eqref{eq.kmessages.Psystem}:
\begin{align}
\label{eq.kmessages.P}
& P(j,l)  - a ( P(j+1,l) - P(j,l)) =
\begin{cases}
- Y_0/2 & 0 \leq l < k, \\
- Y_1 - P(k,j) & l = k,
\end{cases}, \quad j \in \{0, 1, \ldots, k-1\},
\end{align}
which is equivalent to (see \autoref{footnote.P.diffEq}):
\begin{align}
\label{eq.kmessages.P.explicit}
P(j,l) = &
\begin{cases}
\left( (1+1/a)^j - 1 \right) \frac{Y_0}{2} & 0 \leq l < k \\
\left( (1+1/a)^j - 1 \right) \left( Y_1 + \left( (1+1/a)^k - 1 \right) \frac{Y_0}{2} \right)& l = k
\end{cases}, \qquad   (j,l) \in \{0, 1, \ldots, k\}^2.
\end{align}
The above payment rule is identical to the payment rule of exponential price mechanism in Equation \eqref{eq.I.payment}, except when the other player submits the boundary message $k$.  

\autoref{tab.rev} shows the revenue guarantees $\Pi^*_2$ and $\Pi^{\#}_2$ for various prior distributions\footnote{Distribution Beta$(b, c)$ has a p.d.f.\ of $v^{b-1} (1-v)^{c-1} \cdot \frac{\Gamma(b+c)}{\Gamma(b)\Gamma(c)}$ for $v \in [0,1]$.  Beta$(1, 1)$ is of course the uniform distribution. 
} as $\nu \rightarrow 0$ and compares them with the first price auction with two buyers, where the reserve price is chosen to maximize the revenue guarantee in the first price auction \citep*{BergemannBrooksMorris}.  We see that $\Pi_2^*$ and $\Pi_2^{\#}$ are generally very close, though $\Pi_2^{\#}$ is slightly better than $\Pi_2^*$ when the distribution is heavily concentrated among high values.  In these examples $\Pi_2^{\#}$ is always better than the optimal revenue guarantee from first price auction with reserve price.  In \autoref{tab.rev} we also include the optimal revenue guarantees from the exponential price mechanism in \autoref{prop.onebuyer} with one buyer. 


\begin{figure}[h!]
\caption{\label{tab.rev} Revenue guarantees from \autoref{prop.kmessages} ($\Pi^*_2$), \autoref{prop.twobuyers} ($\Pi^\#_2$), first price auction with optimal reserve price \citep*{BergemannBrooksMorris}, and \autoref{prop.onebuyer} ($\Pi^*_1$). }
\begin{center}
 \includegraphics{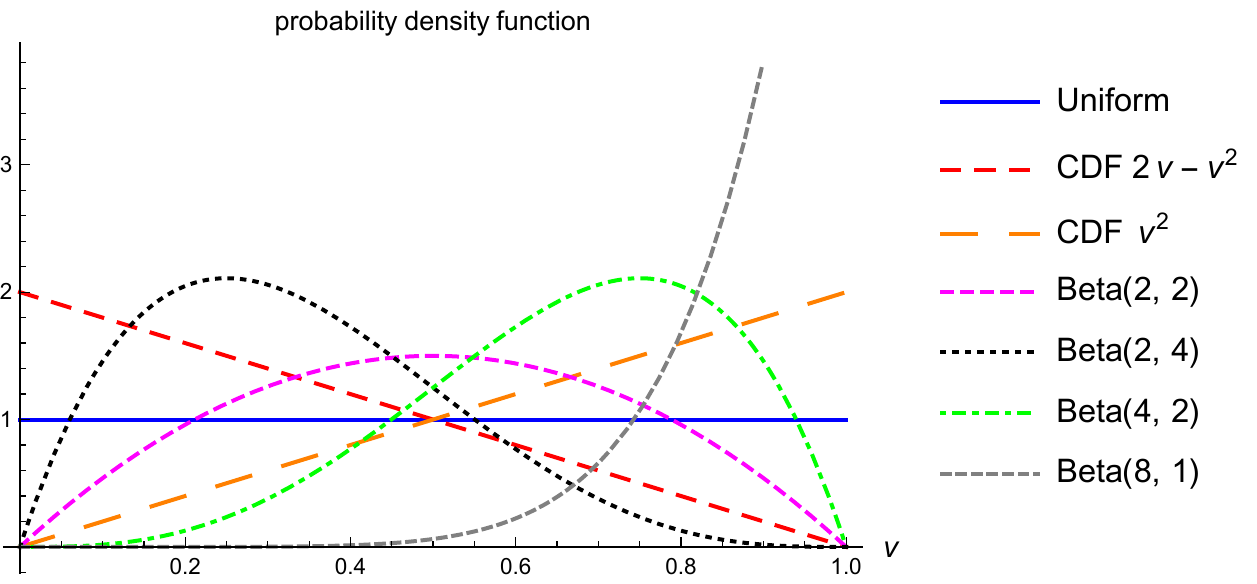}

\vspace{30pt}

\begin{tabular}{|c c | c c | c | c |}
\hline
Prior distribution & Mean & $\Pi_2^{\#}$ & $\Pi^*_2$ & FPA w/ optimal reserve & $\Pi^*_1$ \\
\hline
Uniform & 0.5 & 0.273 & 0.272 & 0.177 & 0.204 \\
\hline
CDF $2v-v^2$ & 0.3333 & 0.166 & 0.166 & 0.102 & 0.120 \\
\hline
CDF $v^2$ & 0.6667 & 0.437 & 0.431 & 0.346 & 0.341 \\
\hline
Beta$(2,2)$ & 0.5 & 0.302 & 0.301 & 0.230 & 0.229 \\
\hline
Beta$(2,4)$ & 0.3333 &  0.188 & 0.188 & 0.139 & 0.140 \\
\hline
Beta$(4,2)$ & 0.6667 & 0.475 & 0.463 & 0.414 & 0.381 \\
\hline
Beta$(8,1)$ & 0.8889 & 0.751 & 0.710 & 0.716 & 0.652 \\
\hline
\end{tabular}
\end{center}
\end{figure}

Finally, consider the following ``wallet game'' information structure: buyer $i$ privately observes signal $s_i$, $i = 1, 2$, where $s_i$ has the uniform distribution on $[0,1]$.  The common value is $v = (s_1+s_2)/2$.  Thus, the prior of the common value is the ``triangle'' distribution.  Given this information structure, the optimal mechanism is a direct mechanism that assigns the good to buyer 1 if $s_1 \geq s_2$ and $s_1+s_2/2 \geq 1/2$; assigns to buyer 2 if $s_2 > s_1$ and $s_2+s_1/2 \geq 1/2$; and does not assign the good otherwise.  (The virtual value of buyer $i$ is $s_i - 1/2 + s_{j}/2$, $j \neq i$.)  The payment rule is given by Myerson's Lemma and makes this mechanism incentive compatible.  The expected revenue of this optimal mechanism is $0.3611$.\footnote{We compute: $2 \cdot \int_{s_1=0}^{1} \int_{s_2=0}^{s_1} \max(s_1+s_2/2 - 1/2, 0) \, d s_2 \, d s_1 = 13/36 \approx 0.3611.$}  Thus, if the prior of the common value is the triangle distribution, $0.3611$ is an (not necessarily tight) upper bound on the revenue guarantee of any mechanism.  We compute for the triangle distribution: $\Pi^*_2 = 0.31094$ and $\Pi^{\#}_2 = 0.31324$, which are $86\%$ of this upper bound.

\section{Conclusion}

We propose a new class of mechanisms (the exponential price mechanisms) to sell a common value good.  The mechanisms are simple and practical, and can guarantee a good revenue over all information structures and equilibria.  The revenue guarantee is provably optimal when there is one buyer, and converges to the full surplus as the number of buyers tends to infinity.  To derive these mechanisms we introduce a linear programming duality approach, which we believe is useful for other robust mechanism design problems, e.g., for studying the revenue guarantee when buyers have both common and private values.

\newpage

\appendix

\noindent {\LARGE\bf Appendix}

\section{Proofs for \autoref{sec.twobuyers}}

\begin{proof}[Proof of \autoref{lemma.kmessages.q}]
By \eqref{eq.kmessages.q.BindingConstr} we have $q(k, k) = 1/2$.  By the second line of \eqref{eq.kmessages.qsystem} this implies that $q(j,k) = j/(2k)$ and $q(k, j) = 1 - j/(2k)$, $j = 0, 1, \ldots, k$.
Then we have
\begin{align}
k - \frac{(k-1) k}{4 k} = \sum_{j=0}^{k-1} q(k, j) = \sum_{j=0}^{k-1} \sum_{l=0}^{k-1} q(l+1, j) - q(l, j) = k^2 q(1,0)
\end{align}
where the last equality follows from the first line of \eqref{eq.kmessages.qsystem}.  Thus, $q(1,0)=(3k + 1)/(4 k^2)$.

%
%
%

\end{proof}

\begin{proof}[Proof of \autoref{lemma.kmessages.P}]
Fix an arbitrary $P$ that satisfies Condition \eqref{eq.kmessages.Psystem}.

From the second line of \eqref{eq.kmessages.Psystem} ($\Rev(v, (j-1, k)) = \Rev(v, (j,k))$), we have
\be
\label{eq.kmessages.P.secondline}
P(j+1, k) - P(j,k) = ( 1 + 1/a ) (P(j,k) - P(j-1, k)) + ( P(k,j) - P(k,j-1) ) / a,
\ee
for $j = 1, 2, \ldots, k-1$.  Equation \eqref{eq.kmessages.P.secondline} implies that
\be
P(j+1,k) - P(j,k) = (1+1/a)^{j} P(1, k) + \sum_{j'=1}^{j} (1+1/a)^{j - j'} (P(k,j') - P(k,j'-1)) / a,
\ee
and as a consequence, for any $j = 0, 1, \ldots, k$:
\be
\label{eq.kmessages.Pjk}
P(j,k) = a ( (1+1/a)^j - 1) P(1,k) + \sum_{j'=1}^{j-1} ( (1+1/a)^{j-j'} - 1 ) (P(k,j') - P(k,j'-1)).
\ee

We claim that
\begin{align}
X(l) & \equiv \sum_{j=1}^{l-1} (1+1/a)^{l-j} (P(l,j) - P(l,j-1)) \nonumber \\
& = P(l,l-1) + a ( (1+1/a)^l - 1 )^2 P(1,0) - (1+1/a)^l P(l,0),
\label{eq.kmessages.P.induction}
\end{align}
for every $l=1, 2, \ldots, k$.  Equation \eqref{eq.kmessages.P.induction} for $l=k$ and Equation \eqref{eq.kmessages.Pjk} together imply Equation \eqref{eq.kmessages.Pkk}, which proves the lemma.

Clearly, \eqref{eq.kmessages.P.induction} is true for $l=1$.  Suppose \eqref{eq.kmessages.P.induction} is true for $l=\kappa < k$ as an induction hypothesis; we prove that this implies \eqref{eq.kmessages.P.induction} is true for $l=\kappa+1$.

From $\Rev(v, (\kappa, j-1)) = \Rev(v, (\kappa, j))$ we have:
\begin{align}
& P(\kappa+1, j) - P(\kappa+1, j-1) \\
= \, & (1+1/a) ( P(\kappa, j) - P(\kappa, j-1) ) + (1 + 1/a) ( P(j, \kappa) - P(j-1, \kappa) ) \nonumber \\
& - ( P(j+1, \kappa) - P(j, \kappa) ), \nonumber
\end{align}
summing the above equation across $j=1, 2, \ldots, \kappa-1$ gives:
\begin{align}
& \sum_{j=1}^{\kappa-1} (1+1/a)^{\kappa +1 -j} (P(\kappa+1, j) - P(\kappa+1, j-1)) \\
= \, & \sum_{j=1}^{\kappa-1} (1+1/a)^{\kappa + 2 -j}  ( P(\kappa, j) - P(\kappa, j-1) ) + \sum_{j=1}^{\kappa-1} (1+1/a)^{\kappa + 2 -j} ( P(j, \kappa) - P(j-1, \kappa) ) \nonumber \\
& - \sum_{j=1}^{\kappa-1} (1+1/a)^{\kappa + 1 -j} ( P(j+1, \kappa) - P(j, \kappa) )  \nonumber \\
= \, & \sum_{j=1}^{\kappa-1} (1+1/a)^{\kappa + 2 -j}  ( P(\kappa, j) - P(\kappa, j-1) ) + (1+1/a)^{\kappa+1} P(1,\kappa) - (1+1/a)^2 (P(\kappa, \kappa) - P(\kappa-1, \kappa)). \nonumber
\end{align}
That is,
\begin{align}
& X(\kappa+1) \\
= \, & \sum_{j=1}^{\kappa-1} (1+1/a)^{\kappa + 2 -j}  ( P(\kappa, j) - P(\kappa, j-1) ) + (1+1/a)^{\kappa+1} P(1,\kappa) \nonumber \\
& - (1+1/a)^2 (P(\kappa, \kappa) - P(\kappa-1, \kappa)) + (1+1/a) (P(\kappa+1, \kappa) - P(\kappa+1, \kappa-1)) \nonumber \\
= \, & (1+1/a)^2 [P(\kappa, \kappa-1) + a ( (1+1/a)^\kappa - 1 )^2 P(1,0) - (1+1/a)^\kappa P(\kappa,0)] + (1+1/a)^{\kappa+1} P(1,\kappa) \nonumber \\
& - (1+1/a)^2 (P(\kappa, \kappa) - P(\kappa-1, \kappa)) + (1+1/a) (P(\kappa+1, \kappa) - P(\kappa+1, \kappa-1)), \nonumber
\end{align}
where in the last equality we have used the induction hypothesis \eqref{eq.kmessages.P.induction} for $l=\kappa$.

From $\Rev(v, (\kappa, 0)) = \Rev(v, (1,0))$ we have $(1+1/a) P(\kappa, 0) - P(1, \kappa) = P(\kappa+1, 0) - 2 P(1,0)$. Therefore, the previous equation is equivalent to:
\begin{align}
 & X(\kappa+1) \\
 = \, & (1+1/a)^2 P(\kappa, \kappa-1) + a (1+1/a)^2 ( (1+1/a)^\kappa - 1 )^2 P(1,0) - (1+1/a)^{\kappa+1} (P(\kappa+1,0) - 2 P(1,0)) \nonumber \\
& - (1+1/a)^2 (P(\kappa, \kappa) - P(\kappa-1, \kappa)) + (1+1/a) (P(\kappa+1, \kappa) - P(\kappa+1, \kappa-1)) \nonumber \\
= \, & (1+1/a)^2 P(\kappa, \kappa-1) + [a (1+1/a)^2 ( (1+1/a)^\kappa - 1 )^2 + 2 (1+1/a)^{\kappa+1}]  P(1,0) - (1+1/a)^{\kappa+1} P(\kappa+1,0) \nonumber \\
& - (1+1/a)^2 (P(\kappa, \kappa) - P(\kappa-1, \kappa)) + (1+1/a) (P(\kappa+1, \kappa) - P(\kappa+1, \kappa-1)) \nonumber
\end{align}

From $\Rev(v, (\kappa, \kappa)) = \Rev(v, (1,0))$ we have $(1+1/a) P(\kappa, \kappa) - P(\kappa+1, \kappa) = - P(1,0)$. Therefore, the previous equation is equivalent to:
\begin{align}
 & X(\kappa+1) \\
= \, & (1+1/a)^2 P(\kappa, \kappa-1) + [a (1+1/a)^2 ( (1+1/a)^\kappa - 1 )^2 + 2 (1+1/a)^{\kappa+1} + (1+1/a) ]  P(1,0) \nonumber \\
& - (1+1/a)^{\kappa+1} P(\kappa+1,0) + (1+1/a)^2  P(\kappa-1, \kappa) - (1+1/a)  P(\kappa+1, \kappa-1). \nonumber
\end{align}
From $\Rev(v, (\kappa-1, \kappa)) = \Rev(v, (1,0))$ we have $(1+1/a) P(\kappa, \kappa-1) + (1+1/a) P(\kappa-1, \kappa)
- P(\kappa+1, \kappa-1) = P(\kappa, \kappa) - 2 P(1,0)$, Therefore, the previous equation is equivalent to:
\begin{align}
 & X(\kappa+1) \\
= \, &  [a (1+1/a)^2 ( (1+1/a)^\kappa - 1 )^2 + 2 (1+1/a)^{\kappa+1} - (1+1/a) ]  P(1,0) \nonumber \\
& - (1+1/a)^{\kappa+1} P(\kappa+1,0) + (1+1/a) P(\kappa, \kappa). \nonumber
\end{align}
Finally, using $(1+1/a) P(\kappa, \kappa) - P(\kappa+1, \kappa) = - P(1,0)$ again we get:
\begin{align}
 & X(\kappa+1) \\
= \, &  [a (1+1/a)^2 ( (1+1/a)^\kappa - 1 )^2 + 2 (1+1/a)^{\kappa+1} - (1+1/a) - 1 ]  P(1,0) \nonumber \\
& - (1+1/a)^{\kappa+1} P(\kappa+1,0) + P(\kappa+1, \kappa). \nonumber
\end{align}

Since $a (1+1/a)^2 ( (1+1/a)^\kappa - 1 )^2 + 2 (1+1/a)^{\kappa+1} - (1+1/a) - 1 = a ( (1+1/a)^{\kappa+1} - 1)^2$, this proves \eqref{eq.kmessages.P.induction} when $l=\kappa+1$.
\end{proof}

\clearpage
\singlespacing
\bibliographystyle{ecta}
\bibliography{../RobustReference}

\begin{thebibliography}{19}
\newcommand{\enquote}[1]{``#1''}
\expandafter\ifx\csname natexlab\endcsname\relax\def\natexlab#1{#1}\fi

\bibitem[\protect\citeauthoryear{Bergemann, Brooks, and Morris}{Bergemann
  et~al.}{2016}]{BergemannBrooksMorris}
\textsc{Bergemann, D., B.~Brooks, and S.~Morris} (2016): \enquote{First Price
  Auctions with General Information Structures: Implications for Bidding and
  Revenue,} Working paper.

\bibitem[\protect\citeauthoryear{Bergemann and Morris}{Bergemann and
  Morris}{2016}]{BergemannMorrisTE}
\textsc{Bergemann, D. and S.~Morris} (2016): \enquote{Bayes Correlated
  Equilibrium and The Comparison of Information Structures in Games,}
  \emph{Theoretical Economics}.

\bibitem[\protect\citeauthoryear{Brooks}{Brooks}{2013}]{Brooks2013}
\textsc{Brooks, B.} (2013): \enquote{Surveying and Selling: Belief and Surplus
  Extraction in Auctions,} Working paper.

\bibitem[\protect\citeauthoryear{Carrasco, {Farinha Luz}, Monteiro, and
  Moreira}{Carrasco et~al.}{2015}]{Carrasco_etal}
\textsc{Carrasco, V., V.~{Farinha Luz}, P.~Monteiro, and H.~Moreira} (2015):
  \enquote{Robust Selling Mechanisms,} Working paper.

\bibitem[\protect\citeauthoryear{Carroll}{Carroll}{2015}]{CarrollAER}
\textsc{Carroll, G.} (2015): \enquote{Robustness and Linear Contracts,}
  \emph{American Economic Review}, 105, 536--63.

\bibitem[\protect\citeauthoryear{Carroll}{Carroll}{2016}]{Carroll_AdverseSelection}
---\hspace{-.1pt}---\hspace{-.1pt}--- (2016): \enquote{Informationally Robust
  Trade and Limits to Contagion,} \emph{Journal of Economic Theory},
  forthcoming.

\bibitem[\protect\citeauthoryear{Chen and Li}{Chen and Li}{2016}]{ChenLi}
\textsc{Chen, Y.-C. and J.~Li} (2016): \enquote{Revisiting the Foundations of
  Dominant-Strategy Mechanisms,} Working paper.

\bibitem[\protect\citeauthoryear{Chung and Ely}{Chung and Ely}{2007}]{ChungEly}
\textsc{Chung, K.-S. and J.~C. Ely} (2007): \enquote{Foundations of
  Dominant-Strategy Mechanisms,} \emph{The Review of Economic Studies}, 74,
  447--476.

\bibitem[\protect\citeauthoryear{Cr\'{e}mer and McLean}{Cr\'{e}mer and
  McLean}{1985}]{CremerMcLean_1985Ecta}
\textsc{Cr\'{e}mer, J. and R.~P. McLean} (1985): \enquote{Full Extraction of
  the Surplus in Bayesian and Dominant Strategy Auctions,} \emph{Econometrica},
  53, 345--361.

\bibitem[\protect\citeauthoryear{Cr\'{e}mer and McLean}{Cr\'{e}mer and
  McLean}{1988}]{CremerMcLean_1988Ecta}
---\hspace{-.1pt}---\hspace{-.1pt}--- (1988): \enquote{Optimal Selling
  Strategies under Uncertainty for a Discriminating Monopolist when Demands are
  Interdependent,} \emph{Econometrica}, 56, 1247--1257.

\bibitem[\protect\citeauthoryear{Engelbrecht-Wiggans, Milgrom, and
  Weber}{Engelbrecht-Wiggans et~al.}{1983}]{Engelbrecht-WiggansMilgromWeber}
\textsc{Engelbrecht-Wiggans, R., P.~R. Milgrom, and R.~J. Weber} (1983):
  \enquote{Competitive bidding and Proprietary Information,} \emph{Journal of
  Mathematical Economics}, 11, 161--169.

\bibitem[\protect\citeauthoryear{Frankel}{Frankel}{2014}]{FrankelAER}
\textsc{Frankel, A.} (2014): \enquote{Aligned Delegation,} \emph{American
  Economic Review}, 104, 66--83.

\bibitem[\protect\citeauthoryear{Hartline and Roughgarden}{Hartline and
  Roughgarden}{2016}]{HartlineRoughgarden_2016}
\textsc{Hartline, J. and T.~Roughgarden} (2016): \enquote{Optimal Platform
  Design,} Working paper.

\bibitem[\protect\citeauthoryear{Myerson}{Myerson}{1997}]{Myerson1997}
\textsc{Myerson, R.~B.} (1997): \enquote{Dual Reduction and Elementary Games,}
  \emph{Games and Economic Behavior}, 21, 183–--202.

\bibitem[\protect\citeauthoryear{Roesler and Szentes}{Roesler and
  Szentes}{2016}]{RoeslerSzentes}
\textsc{Roesler, A.-K. and B.~Szentes} (2016): \enquote{Buyer-Optimal Learning
  and Monopoly Pricing,} Working paper.

\bibitem[\protect\citeauthoryear{Stroock}{Stroock}{2013}]{StroockMarkov}
\textsc{Stroock, D.} (2013): \emph{Introduction to Markov Processes}, Springer,
  2 ed.

\bibitem[\protect\citeauthoryear{Wilson}{Wilson}{1987}]{Wilson1987}
\textsc{Wilson, R.} (1987): \enquote{Game-Theoretic Analyses of Trading
  Processes,} in \emph{Advances in Economic Theory: Fifth World Congress}, ed.
  by T.~Bewley, Cambridge University Press, 33--70.

\bibitem[\protect\citeauthoryear{Yamashita}{Yamashita}{2015}]{YamashitaRestud}
\textsc{Yamashita, T.} (2015): \enquote{Implementation in Weakly Undominated
  Strategies: Optimality of Second-Price Auction and Posted-Price Mechanism,}
  \emph{The Review of Economic Studies}, 82, 1223--1246.

\bibitem[\protect\citeauthoryear{Yamashita}{Yamashita}{2016}]{YamashitaRevenueGuarantee}
---\hspace{-.1pt}---\hspace{-.1pt}--- (2016): \enquote{Revenue Guarantee in
  Auction with Common Prior,} Working paper.

\end{thebibliography}

\end{document}